\documentclass[10pt,letterpaper]{article}
\usepackage[top=0.85in,left=2.75in,footskip=0.75in]{geometry}

\usepackage{amsmath,amssymb}
\usepackage{amsthm,bbm,dsfont}
\usepackage{changepage}
\usepackage[utf8x]{inputenc}
\usepackage{textcomp,marvosym}
\usepackage{cite}
\usepackage{nameref,hyperref}
\usepackage{microtype}
\DisableLigatures[f]{encoding = *, family = * }
\usepackage[table]{xcolor}
\usepackage{array}
\newcolumntype{+}{!{\vrule width 2pt}}
\newlength\savedwidth

\raggedright
\usepackage[aboveskip=1pt,labelfont=bf,labelsep=period,justification=raggedright,singlelinecheck=off]{caption}

\makeatletter
\renewcommand{\@biblabel}[1]{\quad#1.}
\makeatother

\date{August 14, 2017}
\usepackage{lastpage,fancyhdr,graphicx}
\usepackage{epstopdf}
\fancyheadoffset[L]{2.25in}
\fancyfootoffset[L]{2.25in}
\usepackage{booktabs}
\def\l{\left}
\def\r{\right}
\def\mrho{\avg{\rho}}
\def\mgamma{\avg{\gamma}}
\newcommand{\f}{\frac}
\newcommand{\avg}[1]{\l<#1\r>}

\def\fedges{\l< f_{\mathrm{edges}} \r>}
\def\fnodes{\l< f_{\mathrm{nodes}} \r>}
\def\mS{\l< S \r>}
\def\mA{ \l< A \r>}

\DeclareMathOperator*{\argmin}{arg\,min}

\usepackage{xcolor}

\begin{document}
\vspace*{0.2in}

\begin{flushleft}
{\Large
\textbf\newline{Reply \& Supply: Efficient crowdsourcing when workers do more than answer questions}
}
\newline
\\
Thomas C.\ McAndrew\textsuperscript{},
Elizaveta A.\ Guseva\textsuperscript{},
James P.\ Bagrow\textsuperscript{*},
\\
\bigskip
Mathematics \& Statistics, Vermont Complex Systems Center, University of Vermont, Burlington, Vermont, USA
\\
\bigskip

* Corresponding author

E-mail: james.bagrow@uvm.edu
\end{flushleft}

\section*{Abstract}
Crowdsourcing works by distributing many small tasks to large numbers of workers, yet the true potential of crowdsourcing lies in workers doing more than performing simple tasks---they can apply their experience and creativity to provide new and unexpected information to the crowdsourcer.
One such case is when workers not only answer a crowdsourcer's questions but also contribute new questions for subsequent crowd analysis, leading to a growing set of questions.
This growth creates an inherent bias for early questions since a question introduced earlier by a worker can be answered by more subsequent workers than a question introduced later. Here we study how to perform efficient crowdsourcing with such growing question sets. 
By modeling question sets as networks of interrelated questions, we introduce algorithms to help curtail the growth bias by efficiently distributing workers between exploring new questions and addressing current questions.
Experiments and simulations demonstrate that these algorithms can efficiently explore an unbounded set of questions without losing confidence in crowd answers.

\section{Introduction}

Crowdsourcing has emerged as a powerful new paradigm for accomplishing work by using modern communications technology to direct large numbers of people who are available to complete tasks (workers) to others who need large amounts of work to be completed (crowdsourcers)~\cite{howe2006rise,kittur2008crowdsourcing,brabham2008crowdsourcing,kamar2012combining}.
Crowdsourcing often focuses on tasks that are easy for humans to solve, but may be difficult for a computer.
For example, parsing human written text can be a difficult task and optical character recognition systems may be unable to identify all scanned words~\cite{maclean2013identifying,holley2010crowdsourcing,karnin2010crowdsourcing}. 
To address this,
the reCAPTCHA~\cite{von2008recaptcha} system takes scanned images of text which were difficult for computers to recognize and hands them off to Internet workers for recognition.
By having many people individually solve quick and easy tasks, reCAPTCHA is able over time to transcribe massive quantities of text.
Crowdsourcing in general is especially important as a new vehicle for addressing problems of social good~\cite{pickard2011time,tang2011reflecting,naroditskiyVerification2012}.

Deciding on an optimal way to assign particular tasks to workers, and in what order, remains an active area of research.
For many problems, multiple worker responses to a task must be aggregated to determine a final answer~\cite{kamar2012combining} but
often, a budget limits the total crowdsourcing resources available~\cite{li2016crowdsourcing,karger2014budget,tran2013efficient,tran2012efficient}, either due to financial limits when workers are compensated or time constraints where the speed or size of the crowd limits the number of tasks to be performed or questions to be answered.
Most previous work on optimal task assignment  takes the form of a Markov Decision Process (MDP)~\cite{ipeirotis2014quizz,li2016crowdsourcing}.
MDP provides a rigorous mathematical framework to test policies for allocating tasks to workers~\cite{puterman2014markov}. 
Using MDP and other strategies, such as Thompson sampling~\cite{chapelle2011empirical}, methods have been introduced to efficiently aggregate responses from workers, including consideration of which workers are most likely to be well suited for a given question based on their past performance on related questions~\cite{donmez2009efficiently,hung2013evaluation,khattak2011quality,kazai2011worker,whitehill2009whose,ross2010crowdworkers,ipeirotis2010quality,rajan2013crowdcontrol,abraham2013adaptive}.

However, to the best of our knowledge, past research has been limited to the case where a fixed set of tasks need to be accomplished, and the response of a worker to a task is only ever to complete the assigned task. 
In contrast, consider a crowdsourcing problem where workers are able to do more than perform tasks---they may be allowed to propose new questions as well as answers to a given question.
The truest expression of crowdsourcing must incorporate the intuition and experience of workers, who are potentially capable of providing the crowdsourcer with far more actionable information for many problem domains~\cite{bongard2013crowdsourcing,bevelander2014crowdsourcing,salganikWikiSurveyPONE2015}.
While MDP made significant contributions to the design of question assignment algorithms, when the question set is growing due to the crowd, MDP does not naturally account for the hidden state transitions needed to represent newly contributed questions.
The lack of research on algorithms accounting for growing question sets reveals a gap in our abilities to efficiently assign questions to workers.

To this end, we study a type of crowdsourcing problem we term \emph{Reply \& Supply}.
As workers answer a given question (Reply), they are given the opportunity to propose a related question (Supply).
Example applications of Reply \& Supply include:
\begin{itemize} \itemsep=0pt
\item Exploring social networks (\emph{``Are Alice and Bob friends?'' ``Who else is friends with Alice?'' ``With Bob?''})
\item Product recommendations (\emph{``Have you bought a camera and laptop together?'' ``What else would someone buy when buying a camera?''})
\item Image classification (\emph{``Does this photo contain a horse and a mountain?'' ``What else does it contain?''})
\item Causal attribution (\emph{``Do you think `hot weather' causes `violent crime'?'' ``What causes `violent crime'?''})
\item Health informatics: Crowdsourcing patient reports to find connections between co-occurring symptoms, new drug interactions, etc.
(\emph{``Do you suffer from symptom $X$?'' ``What other symptoms do you have?'' ``Do you take drug $Y$?'' ``What other drugs do you take?''})

\end{itemize}
In all these examples, new questions can be built by combining crowd-suggested responses with the components of the original question, leading to the creation of a \emph{network structure} of interrelated questions. 
To explore a social network, for example, if a worker responds that Alice and Bob are friends (Reply) and also proposes that Alice and Carol are friends (Supply), then a new question  (\emph{``Are Alice and Carol friends?''}) is formed that other workers can consider and that links to other questions related to Alice. 
Further, we will show that this network representation naturally generalizes to non-network question sets, and the methods we develop here are fully applicable to both question sets and question nets.

Question networks can be studied with tools from network science that consider the statistical properties governing how theoretical and real-world networks grow and behave~\cite{erdHos1961strength,erdds1959random,barabasi1999emergence,albert2002statistical,watts1998collective,strogatz2001exploring,newman2003structure,newman2004finding}. 
One property, the scale-free or heavy-tailed degree distribution~\cite{barabasi1999emergence}, where most nodes in the network have low degree but some very high-degree nodes do exist, holds in many real-world networks.
How a scale-free network grows over time introduces biases (`first-mover advantage') that are also inherent in a growing crowdsourced experiment.

In brief, this manuscript makes the following contributions:
\begin{enumerate}\itemsep=0pt
    \item The introduction of a growing network of linked questions with an accompanying theoretical analysis;
    \item The use of Thompson sampling to develop crowd-steering algorithms that enable efficient exploration of an evolving set of tasks or questions without losing confidence in answers;
    \item Simulations and real-world crowdsourcing experiments that validate the efficiency and, to some extent, the accuracy of the  crowdsourcing performed under the crowd-steering algorithm.
\end{enumerate}

The rest of this paper is organized as follows: Section~\ref{sec:methods} poses the generic crowdsourcing problem we focus on, analyzes a simple graphical model of how a growing question net is built by a crowd, and uses this model to motivate methods for efficiently assigning questions to workers as the question net grows.
Section~\ref{section:experiments} describes experiments and evaluation metrics to test the proposed theory and methods with both simulated and real-world crowdsourcing tasks.  
Section~\ref{sec:results} presents the results of these experiments and Sec.~\ref{sec:discussion} concludes with a discussion of these results and future work.

\section{Methods}
\label{sec:methods}

Here we introduce a graphical model of a growing question network where questions consider the presence or absence of a relationship between two items (Sec.~\ref{subsec:growingnetmodel}). 
We study the network's properties under a null condition where the crowdsourcer assigns questions to workers randomly without use of a ``steering'' algorithm to provide guidance (Sec.~\ref{subsec:nullmodel}).  
We then use these properties to develop a probability matching algorithm which provides said guidance to the crowdsourcer (Sec.~\ref{subsec:algorithm}).

\subsection{Crowdsourcing growing question networks}
\label{subsec:growingnetmodel}

We model a growing set of questions (or tasks) as a graphs where nodes are items and edges or links represent questions relating pairs of items.
A question network $G = (V,E)$ is composed of a set of nodes $V$ and a set of edges $E$, where $\l|V\r| = N$ and $\l|E\r| = M$.
Edge attributes record the answers given by workers, i.e., associated with each edge is a categorical variable storing the counts of worker responses. 
Those workers may also propose new questions (i.e., new combinations of new or existing items), leading to new nodes and edges. This network model also accommodates non-network question sets, for example by considering each question as a disjoint edge..

As an example of such a network, consider a \textit{synonym proposal task} (SPT) where workers are asked if two words $u$ and $v$ are synonyms. 
The question is the link $(u,v)$ between two items $u$ and $v$ representing those words. 
After replying to the question, the worker may also supply another word $w$ which is a synonym for $u$, for $v$, or for both words. 
This grows the question network by introducing new questions linking items $(u,w)$, or items $(v,w)$, or both ($u,w$) and ($v,w$).
The degree $k_i$ of item $i$ counts the number of questions linking item $i$ to other items.

We focus on cases, such as the SPT, where questions have binary answers, e.g, when workers are asked whether or not a link between two items should exist. 
Edge attributes on links capture the number of `yes' and `no' answers given by workers. 
However, this graph representation is flexible enough to allow edge attributes to contain any number of dimensions and there are no restrictions imposed on how workers propose questions. 
Moreover, this graphical model is capable of representing growing question sets without such relations, for example, a collection of $N$ disjoint questions always containing the response items `True' and `False' only may be a two node, $N$ multi-edge graph. 
While not a particularly meaningful representation, it demonstrates that the algorithms we develop are applicable to general crowdsourcing tasks without modification.
Lastly, one can also extend this model to non-binary, multiple choice questions in several ways, including representing questions as hyperedges in a hypergraph.

\subsection{Null model}
\label{subsec:nullmodel}

We propose a generative null model for a growing question network~\cite{barabasi1999mean,bagrow2008phase}. Beginning from a network with one question, a crowdsourcer randomly chooses existing questions to send to workers also chosen at random. 
Those workers answer the questions and then with some probability also propose new questions. 
We study the properties of the network under these assumptions to motivate the development of a probability matching algorithm that can allow a crowdsourcer to efficiently explore the growing question network.

The network begins (at time $t=0$) with two nodes and one undirected link connecting those nodes, representing a single question considering two items.
Under the null model, every link $(i,j)$ has an associated \textit{innovation rate} $\rho_{ij}$.
The innovation rate for ($i,j$) defines the probability a random worker will introduce a new question into the network when presented with question ($i,j$). 
If she chooses to innovate, the new question may relate to either or both of the items $i$ and $j$ of the original question the worker was given.

Specifically, suppose a random worker is given question $(u,v)$ relating items $u$ and $v$. Under the null model:
\begin{enumerate}\itemsep=0pt

\item The worker answers question $(u,v)$ with probability 1.

\item The worker proposes a new item $w$ to study with probability $\rho_{uv}$:
	\begin{enumerate}\itemsep=0pt
	
		\item $w$ is linked to one of the items of the original question with probability $\gamma_{uv}$. A single new question, either $(u,w)$ or $(v,w)$ chosen uniformly at random, is introduced;
		
        \item otherwise, $w$ links to both items of the original question with probability $1-\gamma_{uv}$. Two new questions, $(u,w)$ and $(v,w)$, are introduced.
        
	\end{enumerate}
	
\item Repeat from (1) with another sampled question and worker until termination.

\end{enumerate}

This model is tractable but quite basic and does not consider many potential details. 
For example, it assumes that while questions may have different innovation rates, workers do not. 
However, for sufficiently large numbers of workers, the average response is always going to be the primary concern, particularly in most crowdsourcing tasks which need to aggregate multiple worker responses to decide upon a final answer for a question.
If it is necessary, a crowdsourcer interested in accounting for variation between workers can propose a statistical model for their features, and then use statistical inference to estimate these worker parameters during crowdsourcing (see also the Discussion).

We now prove several \textit{average} properties of this null model.
Studying the characteristics of the randomly growing, uncontrolled network informs policies that a crowdsourcer may use to manipulate the network (such as the algorithm we develop in Sec.~\ref{subsec:algorithm}).
Many of these results are also informative for non-network growing question sets.

The first theorem describes question growth in the random uncontrolled network.
\newtheorem{theorem}{Theorem}
\begin{theorem}[Rate of question growth]
  The total number of links $M(t)$ as a function of time $t$ is, on average, $ M(t) = \eta t + 1 $ where $\eta = \mrho \l( 2 - \mgamma  \r)$  is termed the \textbf{exploration rate}.
\end{theorem}
\begin{proof}
For the network to grow, a worker must suggest an additional question, which occurs with probability on average $\mrho$ (average of $\rho_{ij}$).
Once the worker commits to a suggestion, one question is added with probability on average $\mgamma$ or two questions are added with probability on average $1 - \mgamma$.
Combining these two possibilities, the total number of questions grows on average over one timestep according to
\begin{equation*}
  M(t+1) = M(t) + \mgamma \mrho + 2\mrho \l(1-\mgamma\r),
\end{equation*}
with initial condition $M(0)=1$ representing the single seed question of the network. 
Making a continuum approximation, this difference equation becomes $M^{\prime}(t) = \mrho \l( 2 - \mgamma  \r)$, which  has solution
\begin{equation}
  M(t) = \eta t + 1,
  \label{linkRate}
\end{equation}
where the \textit{exploration rate} $\eta \equiv \mrho \l( 2 - \mgamma  \r)$ plays an important role in the overall network growth. 
\end{proof}
The number of links grows linearly with a rate $\eta$ that combines the average rates $\mrho$ and $\mgamma$.
Intuitively, the network grows faster if questions are more likely to be innovative (larger $\mrho$), and/or the worker is able to suggest a question for both items at the same time (smaller $\mgamma$).

The solution to the rate equation for question growth can be used to compute the mean number of worker answers per question: 
\begin{theorem}[Mean answer density]\label{thm:meanansdensity}
The mean answer density (number of answers per question) $\avg{A} \to 1/\eta$ as $t\to\infty$.
\end{theorem}
\begin{proof}
The mean number of answers per question is
\begin{equation}
  \avg{A}  = \f{\text{total number of answers}}{\text{total number of questions}}. \label{ansDensDef} 
\end{equation}
At every time step a question in the network accumulates a single answer from a worker.
The denominator of \eqref{ansDensDef} is the solution \eqref{linkRate}, and so the average density of answers per question is
\begin{equation*}
  \avg{A}  = \f{t}{\eta t + 1} = \f{1}{\eta + \f{1}{t}} \to \f{1}{\eta}
\end{equation*}
as $t\to\infty$.
\end{proof}

The mean answer density correlates with the overall uncertainty in the crowdsourcing since there is generally more certainty (but not necessarily correctness) in crowd responses when more workers on average have independently answered questions.
Controlling the answer density, and therefore the certainty, now boils down to controlling the exploration rate $\eta$.
The mean answer density's dependence on $\eta$ also encapsulates an `exploration-exploitation' tradeoff: 
lower $\eta$ leads to higher answer density, but at the cost of less exploration in the network; higher $\eta$ increases the exploration but lowers answer density and makes more uncertainty in the network. 
In this null model, the crowdsourcer does not make choices that can exploit this, but tuning between these poles is a key component of the probability matching algorithm we introduce in Sec.~\ref{subsec:algorithm}.

The previous two theorems govern global properties of random question networks. 
We now turn to properties of individual items within the network to explain the unequal distribution of questions attached to items: 
\begin{theorem}[Rich-get-richer mechanism]
A node $i$ entering the network at time $t_{i}$ will gain degree, on average, as $k_{i}(t) = \f{\eta}{\mrho} \l( \f{1 + \eta t}{1 + \eta t_{i}} \r)^{1/2} \mathcal{H}(t - t_{i})$, where $\mathcal{H}$ is the Heaviside function.  
\end{theorem}
\begin{proof}
An existing item $i$ only gains a question when the crowdsourcer chooses a question attached to $i$ and the worker answering that question proposes a new question involving $i$.
A question $(i,j)$ associated with item $i$ is selected by the crowdsourcer with probability $k_i(t)/M(t)$, where $k_i(t)$ is the degree (number of questions) of $i$ at time $t$.
After the worker answers question $(i,j)$ she must innovate (probability $\mrho$) with an item $w$ that is not already a neighbor of $i$  (and $w \neq i$) and the new question must be $(w,i)$ (probability $\mgamma/2$) or it must be two questions $(w,i)$ and $(w,j)$ (probability $1-\mgamma$). 
If the worker introduces question $(w,j)$ only (probability $\mgamma/2$) then $i$ does not gain a new question and so this possibility does not contribute to $k_i(t)$.
Combining these possibilities together, $k_i(t)$ evolves on average according to
\begin{equation}   
 k_{i}(t) = k_{i}(t-1) + k_{i}\f{\mrho}{M(t-1)} \l( \f{\mgamma}{2} + (1 - \mgamma)\r).
\end{equation}
We approximate and simplify this difference equation as before:
\begin{equation}
  \f{dk_{i}}{dt} = k_{i} \f{\mrho}{M(t)} \l( 1 - \f{\mgamma}{2} \r) = \f{k_{i}}{2} \l( \f{\eta}{\eta t + 1}\r); \; k_i(t_{i}) = \f{\eta}{\mrho},
 \label{eqn:rateEqDegree}
\end{equation}
where $k_i(t_i)$ is the initial degree when item $i$ was introduced at some time $t_i$.
Solving Eq.~\eqref{eqn:rateEqDegree} results in 
\begin{equation}
  k_{i}(t) = \f{\eta}{\mrho} \l( \f{1 + \eta t}{1 + \eta t_{i}} \r)^{1/2} \mathcal{H}(t - t_{i}).
\end{equation}
\end{proof}
We see from this derivation that the rich-get-richer, preferential attachment mechanism~\cite{barabasi1999emergence} is automatic when questions are chosen at random: an item $i$ is more likely to appear in a sampled question the more questions it has, and therefore items with more questions are more likely to gain further questions than other items.
Further, the degree of an item depends critically on two quantities.
The first, the ratio of exploration rate $\eta$ to $\mrho$, equally affects all items in the network.
The second, the time of entry $t_{i}$, dampens the growth of items that enter the network late and increases the growth of earlier items.  
This phenomena is often called the `first mover's advantage', and in the context of crowdsourcing a growing network, items entered earlier in the system accrue more questions than later items.

Using the local estimate of item degree to derive the global degree distribution of the network, we find:
\begin{theorem}[Degree Distribution] The degree distribution of the growing question network
\begin{equation}
P(k(t))  \to  2\l( \f{\eta}{\mrho} \r)^{2} \f{1}{k^{3}}
\end{equation}
as $t\to\infty$.
\end{theorem}

\begin{proof}
Following~\cite{barabasi1999mean}, begin with the cumulative probability distribution of item $i$'s degree:
\begin{align}
  P( k_{i}(t) < k ) &= P \l( \f{\eta}{\mrho} \l( \f{1+\eta t}{1+\eta t_{i}} \r)^{1/2} < k  \r) \notag \\
  &= 1 - P \l( t_{i} <  \eta \l(  \f{1}{k \mrho} \r)^{2} \l( 1+\eta t \r)  - \f{1}{\eta}  \r).
  \label{eqn:part1}
\end{align}
Meanwhile, the entry times $t_{i}$ of items into the network follow a distribution proportional to $\mrho$ uniformly through time:
\begin{equation*}
  P( t_{i} = t ) \propto \mrho,
\end{equation*} 
and, after normalizing, we discover the time of entry follows a uniform distribution.
Referring back to \eqref{eqn:part1} and using the integral definition of a cumulative distribution, 
\begin{equation}
  P( k_{i}(t) < k )= 1 - \f{1}{t} \l( \f{1}{k \mrho} \r)^{2} \eta \l( 1+\eta t \r)  - \f{1}{\eta t}. \label{eqn:part2}
\end{equation}
Lastly, differentiating \eqref{eqn:part2} with respect to $k$ gives the degree distribution:
\begin{equation}
  \f{\partial P(k_{i}(t) < k)}{\partial k} = P(k(t)) = \f{2 \eta \l( \eta t + 1\r) }{ t \mrho^{2} } \f{1}{k^{3}}  \to 2 \l(\f{\eta}{\mrho} \r)^{2} \f{1}{k^{3}}
\end{equation}
as $t \to \infty$.
\end{proof}

Our theoretical analysis is supported by simulations of growing question networks (Fig~\ref{fig:simAgreement}).
We conducted $5,000$ simulations and recorded the degree distribution $P(k)$ and degree $k$ of items across different values of exploration rate $\eta$ and time of item entry $t_{i}$. 
Fig~\ref{fig:simAgreement}(a) validates the slower rate of question accrual for late arriving items, and Fig~\ref{fig:simAgreement}(b) shows the degree distribution's match to theory by the collapse of each curve over multiple values of $\eta$.

\begin{figure}
  \centering
  {\includegraphics[width=\columnwidth]{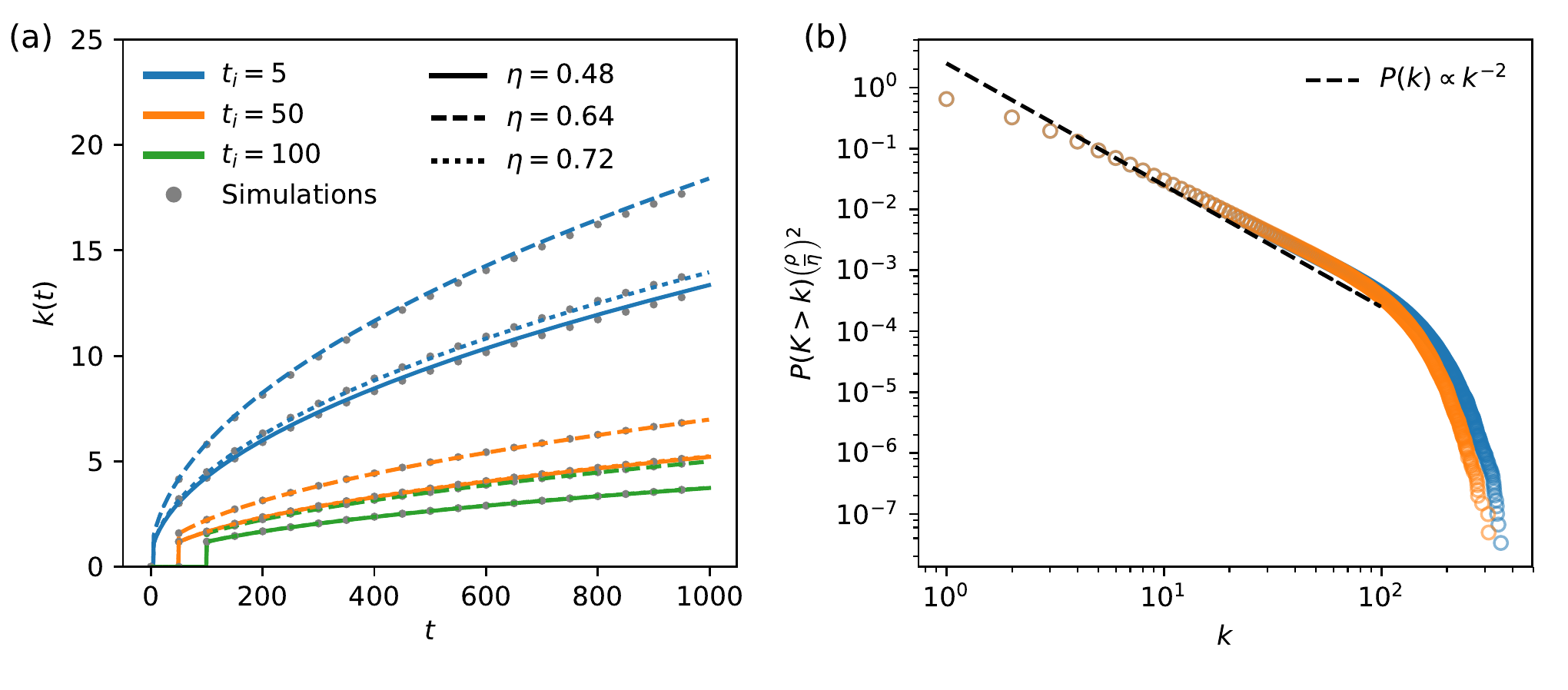}}  
  \caption{Agreement of theoretical predictions of network growth under the null model with simulations for several different choices of parameters. 
  \label{fig:simAgreement} 
  }
\end{figure}

\subsection{Probability matching algorithm for growing question sets and nets}
\label{subsec:algorithm}

Most algorithms for steering workers towards questions choose questions by defining a metric that captures important characteristics in the system. 
For example, algorithms stressing accuracy often build metrics that reward higher numbers of answers for questions, achieving a p-value below a pre-defined threshold, or diminishing the variance of questions.

The framework of probability matching, specifically Thompson sampling~\cite{chapelle2011empirical} (TS), is one of the most powerful ways to efficiently choose from a set of dynamic ``options'' when choices must be made with limited information.
Unlike greedy algorithms, one of the strengths of TS is that its stochastic nature prevents choosing locally optimal questions only.

To Thompson sample from a set of options, one assumes a random variable $X$ which follows a distribution $\varphi(x \mid \theta_i(t))$, where $\theta_i(t)$ is a set of parameters specific to $i$ at time $t$. 
One draws an $x_i(t)$ for each option $i$ and selects the option $j$ with the smallest $x$ (or largest $x$, depending on what $x$ represents), $j = \argmin_i x_i(t)$. 
After option $j$ is played (in our case, the worker's answer is received), the parameters for option $j$ are updated. Often $x$ is a Bernoulli random variable and it is natural for $\varphi$ to be the conjugate Beta distribution with parameters $\alpha,\beta$ which are updated depending on whether $x=0$ or $x=1$.

For specific problems, TS depends on an appropriate reward function. 
In the context of crowdsourcing, one generally cannot verify the accuracy of crowd answers, so the best choice is to reward certainty or consensus. 
If the crowd is consistent in their responses for a given question, then that implies the question is being answered as well as possible under current conditions.
Thus, in contrast to the Bernoulli Bandit problems typically studied with TS, we do not want to reward `yes' answers over `no' answers only. Instead, we want to reward choices that lower the crowdsourcer's measure of uncertainty for questions.

A natural measure of uncertainty for a categorical random variable is the Shannon entropy. However, efficiency is also important to a crowdsourcer.
A yes/no question that has 200 responses which are evenly split is very different than a question with 2 responses which is also evenly split, despite having the same entropy. 
Generally, the crowdsourcer would prefer to assign a worker to the latter question, as there is greater hope of lowering its uncertainty. 

This argument guides us to choosing a metric involving both the total number of answers to a question and how evenly distributed those answers were over the categories of that question.
We introduce a metric called \textit{link bias} ($d$) that is sensitive to the uncertainty of a question, but unlike entropy, also accounts for the total number of answers. 
To begin, the multinomial distribution, with $C-1$ parameters, naturally models the distribution of a categorical question's total number of answers $T$ across $C$ possible answers, and the Dirichlet distribution, conjugate to the multinomial, can estimate the parameters of the multinomial. 
Since we expect no available prior information, a non-informative prior can be used.
In the case of two categories, which we focus on, the Dirichlet distribution reduces to the Beta distribution $(B(\alpha,\beta))$.

To define question uncertainty, we need a reference point.
At a question's peak uncertainty, workers have answered evenly among the question's $(C)$ categories causing an equal proportion of answers per category.
In our binary case $(C=2)$, this corresponds to a proportion of $1/C = 1/2$.
The link bias $d$ transforms the proportion of answers for question $(i,j)$ to the distance from maximum uncertainty with $d \equiv \l| \f{1}{2} - p_{ij}(1)\r|$, where $p_{ij}(1)$ is the fraction of `1' or `yes' or `true' answers. When $p_{ij} \sim B(\alpha,\beta)$, the probability density of $d$ becomes 
\begin{equation}
  \varphi(d\,|\,\alpha,\beta)  = \f{ (1-2d)^{\alpha-1}(1+2d)^{\beta-1} + (1+2d)^{\alpha-1}(1-2d)^{\beta-1}}{B\l(\alpha,\beta\r) 2^{\alpha+\beta-2}},
  \label{eqn:phi}
\end{equation}
where for simplicity the dependence of $\alpha, \beta$ on $(i, j)$ has been suppressed.
Intuitively, a low link bias ($d\approx 0$) occurs when the crowd is evenly split among possible answers, while a high link bias (at most $d = 1/2$) tells us the crowd converged on a single category.

However, the link bias alone may not sufficiently steer the crowdsourcer to choose questions with a lower number of answers. 
If needed, we can combine a preference for sampling questions with few answers, with a preference for questions that are uncertain, by weighting \eqref{eqn:phi} by the current number of answers to define a new `weighted phi' metric $\varphi_{N}$:
\begin{equation}
  \varphi_{N}(d\mid \alpha,\beta) \equiv \f{N_{ij} \varphi( d \;| \alpha, \beta ) }{\sum_{uv \in E} N_{uv}}, 
\end{equation}
where $N_{ij}$ is the total number of answers to question $(i, j)$ at the time of sampling.

Thompson sampling of questions via $\varphi$ or via $\varphi_N$ defines the two probability matching algorithms we propose. 
These algorithms handle growing networks of questions automatically and are fully applicable to problems without graphical relations between questions. 
We will conduct experiments on growing question networks testing the relative performance of both algorithms, and comparing them to other null or control baseline strategies, such as randomly choosing questions.

\section{Experiments}
\label{section:experiments}

We conducted two experiments to test the theoretical analysis and the sampling methods.
For the first experiment, we simulated crowdsourcing of a growing question network with a commonly used benchmarking dataset by superimposing two distinct network structures onto a previously conducted crowdsourcing task~\cite{li2016crowdsourcing}, where questions have been time-ordered to mimic a growing question network, and used this to test three different question sampling algorithms.
For the second experiment, we conducted real-world crowdsourcing using the Mechanical Turk crowdsourcing platform~\cite{kittur2008crowdsourcing}.

\subsection{Experiment 1}
To determine the effectiveness of choosing questions based on link bias, we first performed a five-armed experiment using the Recognizing Textual Entailment (RTE) dataset~\cite{li2016crowdsourcing}, a set of 8,000 binary answers ($0$ or $1$) to $800$ unique questions. 

For simulating question growth, we superimposed graph structures onto the question set to link the 800 questions together.
As mentioned in the introduction, many crowdsourcing problems naturally possess a network structure; here we imposed a structure on the RTE dataset only because it allows us to use the same benchmark dataset that many other researchers have studied.
We built 5,000 Erd{\H o}s-R{\'e}nyi (ER) and Barab{\'a}si-Albert (BA) networks~\cite{newman2010networks}.
These two options represent two extremes of network structure, and were chosen to test question sampling algorithms over different classes of networks.
Briefly, an ER network~\cite{erdHos1961strength} (specifically the $G(n,m)$ formulation) starts with a set of $N$ nodes and $0$ links; a pre-specified number of links $M$ are placed in the network choosing randomly without replacement from all possible $\binom{N}{2}$ pairs of nodes.
In contrast, the BA network~\cite{barabasi1999emergence} starts with $2$ nodes joined by a single link, nodes are added one at a time until all $N$ nodes are placed, and each new node attaches to $m_0$ existing nodes in the network.
New nodes attach to an existing node $i$ with probability $k_{i} / \sum_{n \in N} k_{n}$, a mechanism that is often called \textit{preferential attachment}.

For simulation purposes, each ER network realization must contain exactly $400$ nodes, $800$ links, and be connected.
BA networks are connected by design; we still enforced the same number of nodes and links as the ER networks.
Each simulated crowdsourcing was initialized with one question (a link in the network connecting two corresponding item) chosen at random from the underlying network.
During the simulated crowdsourcing, workers answer a question with a $1$ with probability equal to the proportion of $1$'s observed in the original RTE dataset for that question, otherwise the worker answers $0$.   
Next, and with probability $\mrho$, a new node (item) is introduced into the network by selecting randomly from the unseen neighbors of either $i$ or $j$ within that simulation's graph. (This differs slightly from the analytic null model because there is no $\mgamma$. Instead, two links are formed automatically if the newly introduced item is linked to both $i$ and $j$ in the superimposed network.)
If there are no new items to add corresponding to the selected question, this iteration is undone and the algorithm continues.
All simulations were run with $\mrho = 0.20$ for $6,000$ time steps. 

Simulations were performed independently for each of five arms.
The condition of each arm governs how questions are selected by the simulated crowdsourcer:
\begin{description}\itemsep=0pt

\item[Random:]
The first arm of the experiment had a condition where questions (links) were chosen randomly from the pool of already visited links.

\item[Looping:]
The second arm used a \textit{looping} question sampling algorithm.
The first link that entered the system is answered by a worker, then the second link in the system is given to a worker, then the third link and so on.
When the algorithm reaches the most recent link within the system it starts again from the oldest link.

\item[Binomial sampling:]
This strategy selects questions $(i, j)$ based on p-values for a two-sided binomial test that the proportion $p_{ij}(1)$ is significantly different from $1/2$. If the p-value of this exact test is small, then it is likely the crowd has already reach consensus on that question and it is not worthwhile to sample that question further.
The sampled question was chosen randomly from the set of questions which have a p-value $> 0.2$ and which have received fewer than 10 answers (at the time of sampling)

\item[Thompson sampling with $\varphi$:]
The fourth arm uses Thompson sampling to select links based on link bias $(\varphi)$.

\item[Thompson sampling with $\varphi_N$:]
As in the fourth arm but links are Thompson sampled with $\varphi_N$ instead of $\varphi$.
\end{description}
This experiment can demonstrate the strengths and weaknesses of selecting links based on these different sampling strategies, and, because it is synthetic, many trials can be conducted while avoiding the costs associated with a new crowdsourcing experiment. 
Results of Experiment 1 are presented in Sec.~\ref{sec:results}.

\subsection{Experiment 2: Synonym Proposal Task}
\label{subsec:exp2}

This three-armed experiment created new question networks grown from a single seed question (link), and evaluated the $\varphi_N$-based Thompson sampling and Binomial sampling versus Random sampling. 
We paid US-based workers on Amazon's Mechanical Turk crowdsourcing platform~\cite{kittur2008crowdsourcing,buhrmester2011amazon} to participate in a \textbf{synonym validation and proposal experiment}.
Synonymy proposal is a good test application for the question sampling algorithms we study because workers can easily understand the question and are capable of proposing new questions (by suggesting new synonyms).
Of course, data on synonymy relations are available in lexical resources such as WordNet~\cite{miller1995wordnet}, which we used in this specific task for assessing the accuracy of proposed synonyms (see below), but our primary goal with this experiment is not crowdsourcing a new thesaurus but testing the different question sampling strategies.

In Experiment 2, each worker completes synonymy tasks at a compensation of \$0.08 USD per task.
Each synonymy task gives a pair of words to a worker and asks whether or not they are synonyms.
After a worker answer either `yes' or `no', we allow the worker to suggest additional synonyms for each word of the given pair, or a single synonym associated with the combined word pair. 
A screenshot of the  web form used for this task is shown in Fig~\ref{fig:SPTinterface}.

\begin{figure}[t!]
  \centering 
  {\includegraphics[width=0.5\columnwidth]{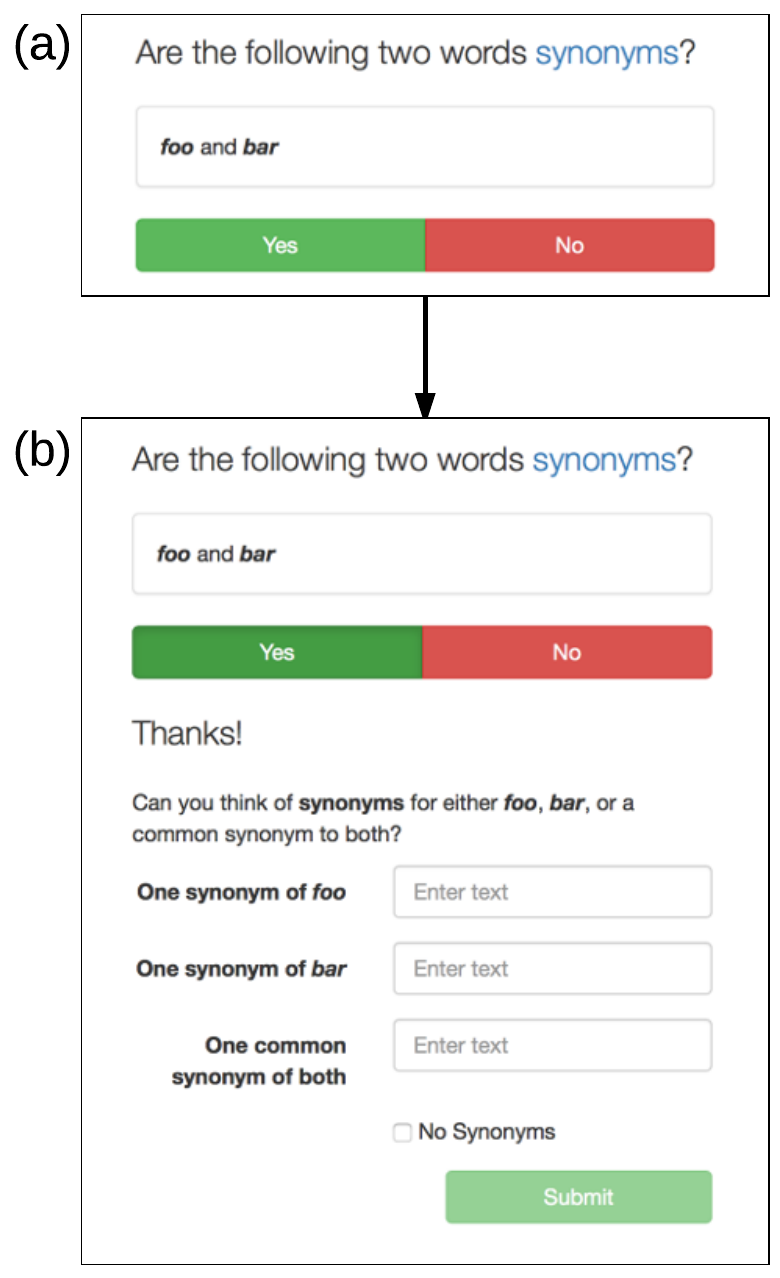}}
  \caption{Screenshots of the Mechanical Turk web interface for the synonymy proposal task (Experiment 2).
  After replying `yes' or `no' (a), the form expands for the worker to supply new potential synonym pairs (b).
  \label{fig:SPTinterface}}
\end{figure}

Three independent crowdsource networks were built, one for each arm. 
All three networks began with the same seed question (the word pair \emph{patriotic}, \emph{person}). 
All other word pairs were proposed by the crowd.
The question sampling algorithms draw from all previous worker answers and suggested questions within their respective arms to deliver a question to the next queued worker.
The first arm (Random sampling) chooses links using the same methodology as the random arm from Experiment 1, which also closely matches the null model we studied (Sec.~\ref{subsec:nullmodel}).
The second arm (Binomial sampling) selects links according to the Binomial sampling algorithm introduced in Experiment 1.
Lastly, the third arm (Thompson sampling) selects links
according to Thompson sampling of $\varphi_N$.
Results for Experiment 2 are presented in Sec.~\ref{sec:results}.

\subsection{Evaluation metrics}
\label{subsec:evalmetrics}

For the first experiment, we measure five attributes across the simulated crowdsourcings to compare the different question sampling algorithms.
At each time step $t$, for each simulated network we record network properties $f_{\mathrm{nodes}}$, the fraction of items, and $f_{\mathrm{edges}}$, the fraction of questions:
\begin{equation}
  f_{\mathrm{nodes}} =  \frac{\left|V(t)\right| }{\left| V(\infty) \right|};  ~ f_{\mathrm{edges}} =  \frac{\left|E(t)\right| }{\left| E(\infty) \right|} ,
\end{equation}
where 
$V(t)$ is the set of items at time $t$,  
$V(\infty)$ is the set of all items at the end of the experiment, 
$E(t)$ is the set of questions at time $t$, and 
$E(\infty)$ is the set of all questions at the end of the experiment.

Next, we record the entropy $S$ and link bias $d$, averaged over all currently visible questions, to quantify uncertainty in the network:
\begin{align}
  \avg{S} = - \frac{1}{\left|E(t)\right|} \sum_{ij \in E(t)}  \sum_{x \in \{0,1\}} p_{ij}(x)\log_{2}p_{ij}(x) 
\end{align}
and
\begin{align}
  \avg{d} = \frac{1}{\left|E(t)\right|}\sum_{ij \in E(t)} { \l| \f{1}{2} - p_{ij}(1) \r| },
\end{align}
where $p_{ij}(x)$ is the (Laplace-smoothed) fraction of binary answers of $x$ for question $(i,j)$ (at time $t$). 

The final evaluation metric, mean answer density, measures how many answers are given per question in a particular network (see also Thm.~\ref{thm:meanansdensity}):
\begin{equation}
  \avg{A} =\frac{1}{\left|E(t)\right|} \sum_{ij \in E(t)} \sum_{x \in \{0,1\}} N_{ij}(x) ,
\end{equation}
where the $N_{ij}(x)$ represents the count of answer $x$ for question $(i,j)$ (at time $t$).

\subsubsection*{Validating proposed synonyms}
\label{subsec:synonymMeasures}

A factor that motivated us to choose the synonym proposal task as our crowdsourcing example is that synonym proposal can, in principle, be validated.
Therefore, we will measure both crowd consensus (measured by $\mS$ or $\left<d\right>$) and, as best we can, if Experiment 2's crowdsourcing algorithms lead to different quality rates of synonyms---are we trading off quality for efficiency?

However, measuring synonymy from natural language text is challenging. 
In principle, all that is needed is a complete thesaurus, meaning a complete lookup table of all words and all their synonyms, perhaps with weights denoting the degree of relatedness between a word and its synonym and accounting for all possible contexts in which those words may appear.
However, without such an exhaustive resource, it can be challenging to determine synonyms, especially when workers may introduce typos, may propose different forms (\emph{runs, running, ran}) of the same root lemma (\emph{run}), or they may propose a multi-word phrase (MWP) which may have a synonymous meaning but where such a meaning is difficult to determine computationally. 

Given the challenges of measuring synonymy, we applied two measures to the synonym word pairs $(u,v)$ proposed by workers during the crowdsourcing experiments:

\begin{description}
\item[Shared WordNet lemmas] 
The first measure starts by determining for each word $w$ the set of all forms of all its synonym lemmas as encoded in WordNet~\cite{miller1995wordnet}:
\begin{equation}
L(w) = \bigcup_{s \in \mathrm{synsets}(w)} \bigcup_{\ell \in \mathrm{lemmas}(s)} \ell,
\end{equation}
where $\mathrm{synsets}(w)$ is the set of all synonym forms stored in WordNet (we merge sets across parts-of-speech and take $\mathrm{synsets}(w)=\emptyset$ if $w$ is not present in WordNet).
We then say that the two words $u$ and $v$ are synonyms if they share at least one lemma, i.e.\  that  $\left| L(u) \cap L(v) \right| > 0$, otherwise they are not synonyms. 
This is a relatively  strict test, and fails to account for many MWPs and natural language concerns such as misspellings, so we expect many $(u,v)$ pairs that workers  deem synonyms to be missed by this measure and therefore the actual proportion of synonymous word pairs may be much higher.

\item[Word vector similarity]
The second measure we employ is based on the meanings encoded by the ``word2vec'' word embedding algorithm~\cite{mikolov2013distributed}. 
Word2vec uses a neural network model to learn low-dimensional vector representations of words based on their contextual co-occurrence patterns over a very large  text corpus. 
Supported by the distributional hypothesis~\cite{harris1954distributional}, the contexts encoded in these  vectors are then considered to capture to some extent the meanings and relationships of these words such as, for example, analogous relationships (\emph{Berlin is to Germany as Paris is to France}). 
Given a pre-trained set of 300-dimensional vectors trained on a 100B word corpus taken from Google News, we define the similarity between two words (or MWPs, if the MWPs are present in the vector data) $u$ and $v$ as their cosine similarity: 
\begin{equation}
\text{similarity}(u,v) = {\mathbf{u} \cdot \mathbf{v} \over \|\mathbf{u}\| \|\mathbf{v}\|},
\end{equation}
where $\mathbf{w}$ represents the associated word vector for word or MWP $w$.
If either $u$ or $v$ is not present in the word2vec vector data, we exclude that pair from our analysis (this occurred in Experiment 2 for approximately 17.9\% of crowd-proposed word pairs for the Random sampling experiment, 19.8\% for Binomial sampling, and 10.9\% for Thompson sampling). 
\end{description}

\section{Results}
\label{sec:results}

\subsection*{Experiment 1}

Fig~\ref{fig:evalMetrics} displays the five evaluation metrics associated with Experiment $1$, averaged over the $5,000$ ER and BA networks.
(For simulated Binomial sampling only, note that we required questions to have fewer than 30 answers at the time of sampling, not 10  as discussed previously, to provide more simulation statistics.)
The Binomial sampling and $\varphi_{N}$ Weighted Thompson sampling algorithms outperformed all others in exploration metrics across ER and BA networks.
Both methods explored more of the network, and faster, than other methods, as evidenced by $\fedges$ and $\fnodes$.
Weighted Thompson sampling performed best at minimizing the uncertainty of answers, as measured by lower entropy $\avg{S}$ and higher link bias $\avg{d}$. 
In contrast, Binomial and Thompson sampling $\varphi$ were inconsistent for these two metrics.
Lastly, Binomial and $\varphi_{N}$ Weighted Thompson sampling also required fewer answers than other algorithms (lower $\mA$).

\begin{figure}[t!]
  \centering
  \includegraphics[width=0.9\columnwidth]{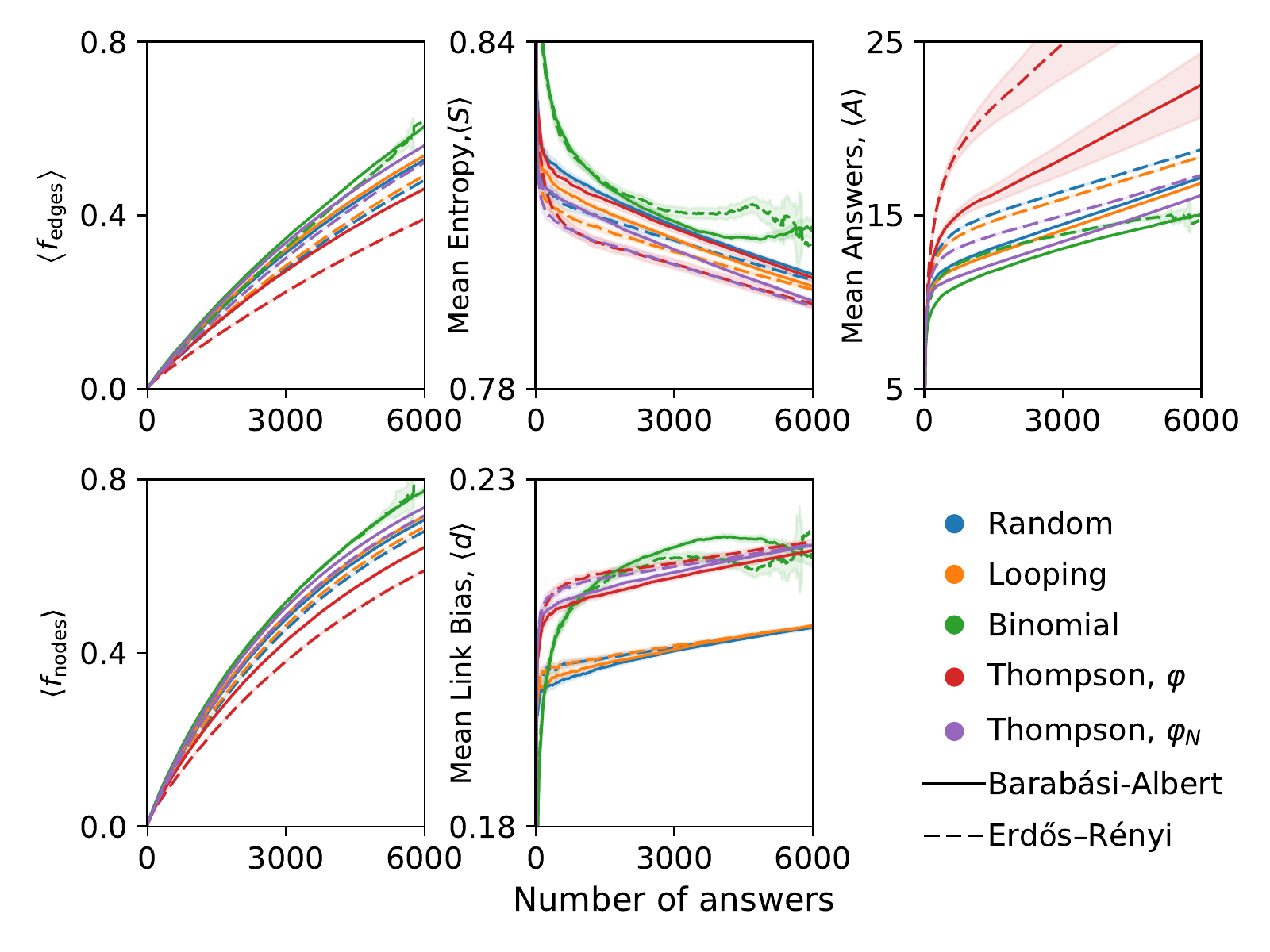}
  \caption{Experiment 1's evaluation metrics for five different question sampling algorithms. 
 \label{fig:evalMetrics}}
\end{figure}

Binomial sampling slightly outperformed $\varphi_{N}$ Weighted Thompson sampling in many metrics. 
However, Binomial sampling has a distinct drawback: the thresholds used to sample questions may lead to a situation where no questions meet its sampling criteria. 
This is visible in the simulation curves, which are quite noisy due to individual simulations which terminated too early. 
Of course, this can be fixed by any of several means, such as falling back to random sampling when no questions meet the criteria, or tuning the cutoffs used in Binomial sampling. But Thompson sampling avoids these complexities entirely.

The overall performance of Binomial sampling and $\varphi_{N}$-based Thompson sampling in simulated crowdsourcing nominates them as candidate algorithms for Experiment 2's real crowdsourcing.

\subsection*{Experiment 2}

Fig~\ref{fig:experimentNets} shows the constructed networks for each arm of the Synonym Proposal Task (the task is described in Fig~\ref{fig:SPTinterface}).
Qualitatively, all three networks appeared similar.
Quantitatively, (Tab.~\ref{tab:table}) both Binomial and Thompson sampling were able to explore more of the network (discovering more items and questions) than Random sampling with more efficiency (lower mean number of answers $\avg{A}$).
The explored networks appeared similar by a number of network metrics, although the network generated via Binomial sampling has a lower average degree and higher average shortest path length. 
Lastly, Binomial and Thompson sampling were comparable to Random sampling in crowd consensus on individual answers, having similar levels of entropy and link bias. 
Both of these statistics measured how skewed the worker answers were in favor of `yes, they are synonyms' or `no, they are not synonyms'.  

Taken together, both Binomial and Thompson sampling maintained a comparable level of certainty (measured by consensus or consistency in worker responses) in the network with fewer answers needed on average than Random sampling.

\begin{figure}[t!]
  \centering
  \includegraphics[width=\textwidth]{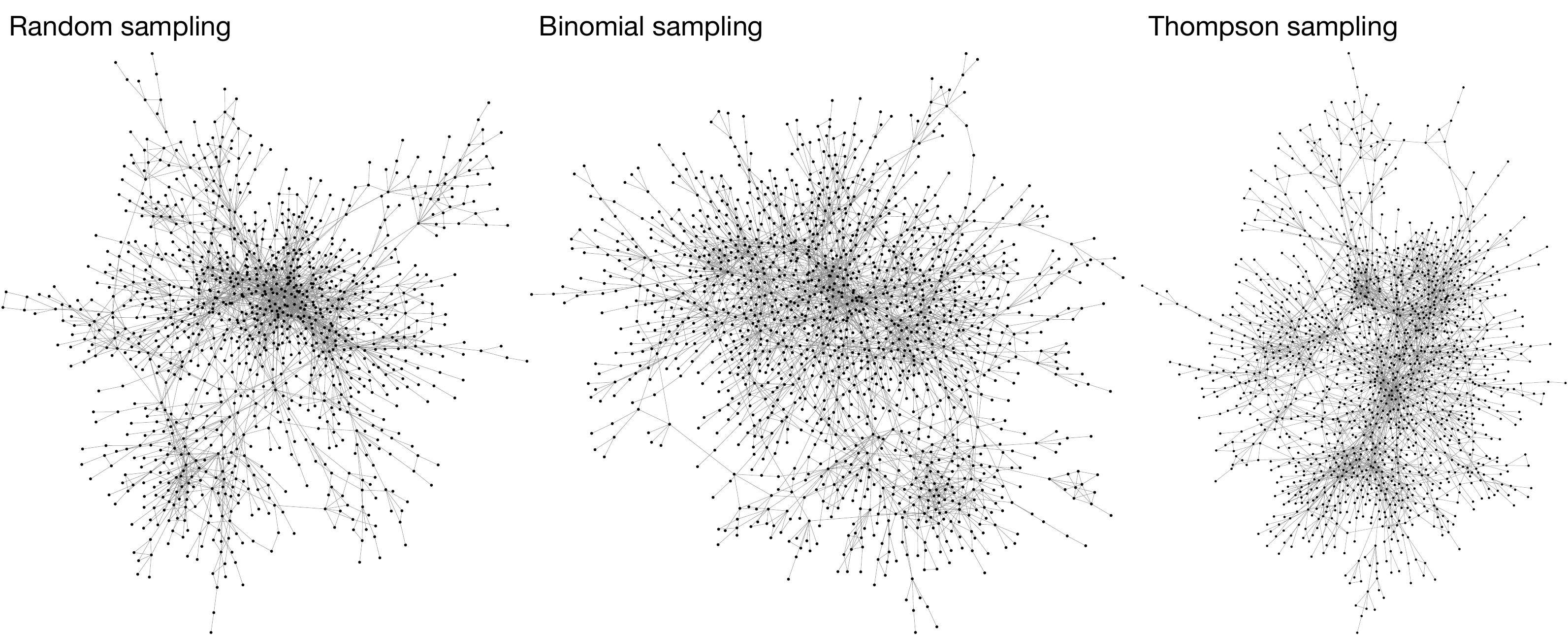}
  \caption{Comparison of question networks for the synonymy proposal task under random sampling, binomial sampling, and $\varphi_{N}$ Thompson sampling.\label{fig:experimentNets}}
\end{figure}

\begin{table}[t!]
\centering
\begin{tabular}{lrrr}
\toprule
{} &    Random &  Binomial &  Thompson \\
\midrule
$N$(items)                   		&  1134 &  1537 &  1509 \\
$N$(questions)               		&  2413 &  2887 &  3020 \\
$N$(responses)               		&  5043 &  4993 &  5034 \\
$\avg{A}$                   		&     2.090 &     1.729 &     1.667 \\[\defaultaddspace]
Average degree, $\avg{k}$    		&     4.256 &     3.757 &     4.003 \\
Clustering coefficient, $\avg{CC}$  &     0.265 &     0.220 &     0.243 \\
Eccentricity, $\avg{e}$             &    10.304 &    11.247 &    11.840 \\
Shortest path length, $\avg{\ell}$  &     5.732 &     6.346 &     5.982 \\[\defaultaddspace]
Entropy, $\avg{S}$                  &     0.560 &     0.553 &     0.551 \\
Link bias, $\avg{d}$                &     0.361 &     0.371 &     0.390 \\
\bottomrule
\end{tabular}
\caption{Summary statistics for the three arms of Experiment 2. 
Both Binomial and Thompson sampling are more efficient than Random sampling (lower $\avg{A}$) without losing the crowd's average consensus on answers, measured by $\avg{S}$ and $\avg{d}$.
\label{tab:table}
}
\end{table}

To further understand the answer density of the different sampling methods, we computed the distribution of the number of answers $N_{ij}$ to question $(i,j)$ in Fig~\ref{fig:numAnswersExperiment2}. 
Here Random sampling clearly separated from the other two sampling strategies, and Random sampling ended with more questions with more answers than the other sampling strategies.
We also note that all three arms finished with many questions with few answers: approximately 50\% of questions at the end of the experiment had a single answer. 
We discuss this further in Sec.~\ref{sec:discussion}.

\begin{figure}[t!]
\centering
\includegraphics[width=0.7\textwidth]{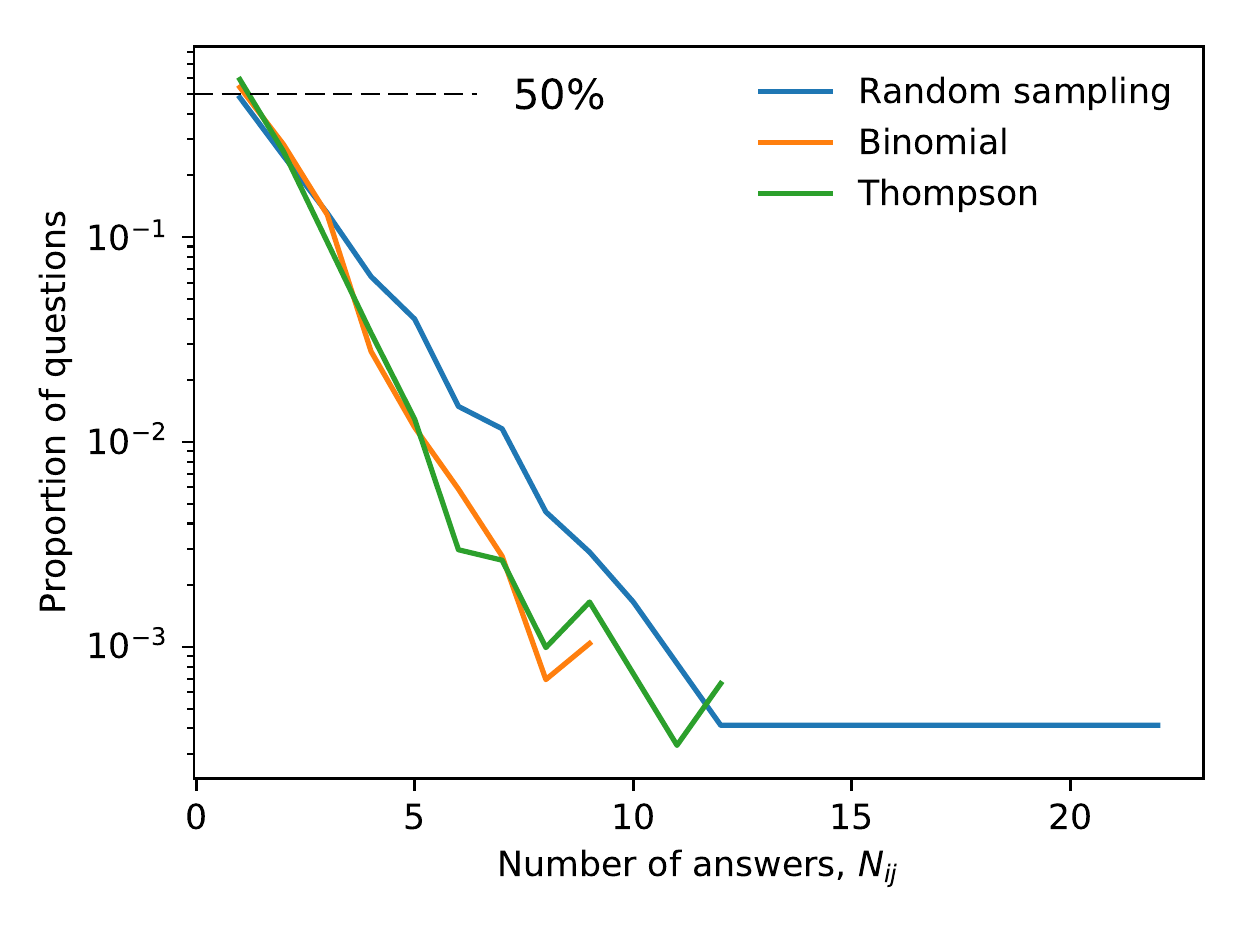}
\caption{The distributions of the total number of answers per question at the end of crowdsourcing, for each arm of Experiment 2. 
The efficiency of Binomial and Thompson sampling compared with Random sampling is clear. 
In all arms, approximately 50\% of proposed questions are answered only once.
\label{fig:numAnswersExperiment2}}
\end{figure}

Next, we examined the synonym ``quality'' of the SPTs, using the synonymy measures introduced in Sec.~\ref{subsec:synonymMeasures}.
We limited these calculations to proposed word pairs examined by at least three crowd workers to ensure sufficient answers from the crowd.
Fig~\ref{fig:validateSPT}(a) shows the proportion of word pairs that share at least one WordNet lemma: 
Both Binomial and Thompson sampling have slightly higher proportions than Random sampling, at over 12\% compared with approximately 11\%. 
This indicates that quality was not lost when using a more efficient sampling strategy.
Of course, 11--12\% of word pairs sharing a lemma seems low, but recall that shared lemmas is a very strict measure that is likely to miss many synonymous word pairs and so we do not conclude that the majority of the crowd answers are ``wrong.''
Furthermore, to better understand the shared lemma proportion, we constructed a randomized control by shuffling the word pairs (preserving the total frequencies of individual words) and re-measured the proportion of shared lemmas. 
We found a significant drop in the proportion to approximately 1\% (the error bars on these proportions are shown in Fig~\ref{fig:validateSPT}(a)  but are quite small).

\begin{figure}[t]
  \centering
  {\includegraphics[width=\columnwidth]{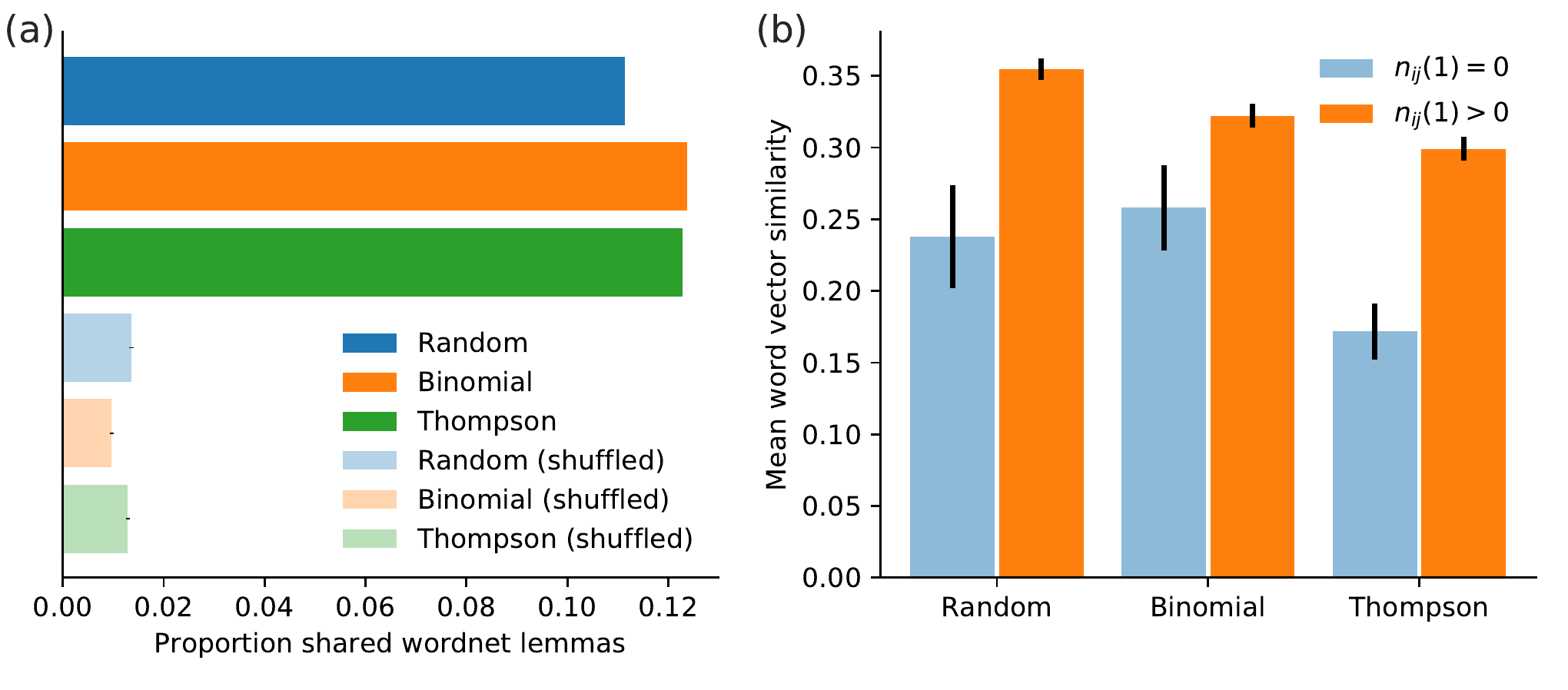}}
  \caption[]{Measures of synonymy for Experiment 2's crowdsourced word pairs.
  Synonymy for proposed word pairs was estimated using (a) shared WordNet lemmas, (b) cosine similarity between word2vec word embedding vectors (see Sec.~\ref{subsec:evalmetrics}).
Approximately 11-12\% of crowdsourced word pairs share one or more WordNet lemmas (a strict measure of synonymy), and Binomial and Thompson sampling achieved slightly higher rates than Random sampling. 
As a control, the  word pairs proposed by the crowd were randomized, and the proportion of word pairs with shared lemmas dropped significantly.
Meanwhile, regardless of sampling algorithm, the word pairs proposed by the crowd also had significantly higher word vector similarity when at least one member of the crowd agreed that the pair were synonymous ($n_{ij}(1)>0$), as opposed to no members agreeing the pair were synonymous ($n_{ij}(1)=0$). 
This further underscores the estimated quality of the proposed questions and answers and that Binomial and Thompson sampling methods do not appear to trade off quality for efficiency. 
(To avoid ambiguous answers, we considered word pairs that received at least three answers from workers in these calculations, and panel (a) considers those word pairs with $n_{ij}(1)>0$.)
  \label{fig:validateSPT}}
\end{figure}

Likewise, Fig~\ref{fig:validateSPT}(b) shows word vector similarities for the three sampling methods, decomposed into word pairs where at least one worker agreed they were synonyms versus no workers agreeing they were synonyms. 
The crowd-proposed word pairs flagged as synonyms had similarities significantly higher than those not flagged as synonyms. 
There is a small drop in vector similarity for Binomial and Thompson sampling compared with Random sampling, likely balancing out the small increase in WordNet shared lemma proportion shown in Fig~\ref{fig:validateSPT}(a). 
We conclude that overall there is no loss in quality, at least as indicated by these measures, when using more efficient sampling algorithms.

Taken together, while we only have one crowdsourcing realization for each arm, it is reasonable to conclude from Experiment 2 that both Binomial sampling and Thompson sampling achieved much higher rates of exploration (more items) and greater efficiency (fewer answers per question) than Random sampling without losing confidence or accuracy in question responses.

\section{Discussion}
\label{sec:discussion}

We studied the problem of efficient assignment of crowdsourcing tasks to workers when those workers are also able to propose tasks themselves. Using workers to contribute new tasks and not merely perform predetermined tasks helps unlock the true potential of crowdsourcing.
We formulated a growing question network model for this problem, prove theoretical properties of this system, and developed and validated sampling algorithms that can guide workers to grow the network efficiently, while only sacrificing at most minimal confidence in their responses.

Modeling the evolution of the uncontrolled question network teaches us how to better design crowdsourcing policies. 
For example, by monitoring the innovation rate ($\rho$) and exploration rate ($\eta$) of the growing question network, a crowdsourcer may be able to better and more efficiently control the question network as it grows.
At the same time, the rich-get-richer growth of items (older items are attached to a larger fraction of questions), implies that crowdsourcers should pay special attention to the newest items entering the network, to balance out the inherent bias in favor of older items.

Thompson sampling is fast, easy to implement, and flexible enough to capture the preferences of different crowdsourcers, but it is only one potential policy for question selection.  
More rigorous question selection techniques can be implemented which may outperform the proposed techniques, but with potentially more restrictions.
The Thompson sampling algorithms proposed here work for both question nets but also non-network question sets, and can naturally accommodate both growing and static questions sets and nets.
Further, statistical inference of question parameters and worker features~\cite{tran2012efficient}, based on extensions of the null model analyzed in Sec.~\ref{subsec:algorithm}, can be used by the crowdsourcer to better pair workers with questions.

There remains considerable room for improvement. For example, in Fig~\ref{fig:numAnswersExperiment2}, approximately 50\% of questions in Experiment 2 received a single answer, regardless of arm. This means that even with the current algorithms the crowd is still supplying an inordinate amount of questions that are being left mostly unconsidered.
Of course, some of this may be unavoidable; if there is too much Supply, then the crowd will invariably fall behind. 
This is further compounded by the inherent bias in favor of older questions.
Thompson and Binomial sampling helped curtail this ``first-mover-advantage'' bias in the growing network but did not necessarily eliminate it.
This is the fundamental challenge (and appeal) of this crowdsourcing problem, and more work focused on these issues is needed.

In the future, we will address more detailed schemes for question selection.
Questions that contain more than a binary (true/false) response should be further investigated, although the only adaptation of the Thompson sampling algorithm is in the choice of metric to Thompson sample from.
Different network structures may arise for different crowdsourcing problems, and assessing the accuracy of the network inferred by the crowdsourcing, and not necessarily the accuracy of individual links, will also be investigated. 
These and many other interesting and important questions remain in the new problem of crowdsourcing with growing question nets and sets.

\section{Acknowledgments}
We thank M.~Wagy and J.~Bongard for useful comments and gratefully acknowledge the resources provided by the Vermont Advanced Computing Core.
This material is based upon work supported by the National Science Foundation under Grant No.\ IIS-1447634.

\end{document}